\documentclass[letterpaper, 10pt, conference,dvipsname]{ieeeconf}
\usepackage{cite}
\usepackage[dvipsnames]{xcolor}
\IEEEoverridecommandlockouts                              
\let\labelindent\relax
\usepackage{enumitem}
\usepackage{psfrag,amsthm,amsmath,amsfonts,verbatim, mathtools,enumerate,algorithm,algorithmicx, dsfont}
\usepackage{diagbox}
\usepackage{subfig}
\usepackage{times,graphicx,algpseudocode,bm,hyperref,tikz}
\allowdisplaybreaks[2]

\newcommand{\reals}{{\mathbb{R}}}

\newcommand{\norm}[1]{\left\lVert#1\right\rVert}
\newcommand{\mnorm}[1]{{\left\vert\kern-0.25ex\left\vert\kern-0.25ex\left\vert #1 
    \right\vert\kern-0.25ex\right\vert\kern-0.25ex\right\vert}}

\newcommand{\mc}{\mathcal}

\newcommand{\A}{\mathcal{A}}

\newtheorem{definition}{Definition} 
\newtheorem{theorem}{Theorem}

\newtheorem{assumption}{Assumption}

\algnewcommand\algorithmicforeach{\textbf{for each}}
\algdef{S}[FOR]{ForEach}[1]{\algorithmicforeach\ #1\ \algorithmicdo}

\begin{document}
\title{\Large \bf Tolling for Constraint Satisfaction in Markov Decision Process Congestion Games}
\author{Sarah H. Q. Li$^{1}$, Yue Yu$^{1}$, Daniel Calderone$^{2}$, Lillian Ratliff$^{2}$, Beh\c cet\ A\c c\i kme\c se$^{1}$
\thanks{*This work was is supported by NSF award CNS-1736582.}
\thanks{$^{1}$Authors are with the William E. Boeing  Department of Aeronautics and Astronautics, University of Washington, Seattle. 
        {\tt\small sarahli@uw.edu}
        {\tt\small yueyu@uw.edu}
        {\tt\small behcet@uw.edu}}%
\thanks{$^{2}$Authors are with the Department of Electrical Engineering, University of Washington, Seattle.
        {\tt\small ratliffl@uw.edu }
        {\tt\small djcal@uw.edu }}%
}

\maketitle
\begin{abstract}
    Markov Decision Process (MDP) congestion game is an extension of classic congestion games, where a continuous population of selfish agents  each solve a Markov decision processes with congestion: the payoff of a strategy decreases as more population uses it. We draw parallels between key concepts from capacitated congestion games and MDPs. In particular, we show that the population mass constraints in MDP congestion games are equivalent to imposing tolls/incentives on the reward function, which can be utilized by a social planner to achieve auxiliary objectives. We demonstrate such methods on a simulated Seattle ride-share model, where tolls and incentives are enforced for two distinct objectives: to guarantee minimum driver density in downtown Seattle, and to shift the game equilibrium towards a maximum social output.     
\end{abstract}
\section{Introduction}\label{sec:introduction}
We consider a class of non-cooperative games, \emph{Markov decision process congestion games} (MDPCG)~\cite{calderone2017markov,calderone2017infinite}, which combine features of classic \emph{nonatomic routing games} \cite{wardrop1952,beckmann1952continuous,patriksson2015traffic}---i.e.~games where a continuous population of agents each solve a shortest path problem---and \emph{stochastic games} \cite{shapley1953stochastic,mertens1981stochastic}---i.e.~games where each agent solves a Markov decision process (MDP). In MDP congestion games, similar to \emph{mean field games with congestion effects}~\cite{lasry2007mean,gueant2015existence}, a continuous population of selfish agents each solve an MDP with congestion effects on its state-action rewards: the payoff of a strategy decreases as more population mass chooses it. An equilibrium concept for MDPCG's akin to the \emph{Wardrop equilibrium}~\cite{wardrop1952} for routing games was introduced in~\cite{calderone2017markov}.

In this paper, we consider \emph{modifying} MDPCG's game rewards to enforce artificial state constraints that may arise from a system level. For example, in a traffic network with selfish users, tolls can be used to lower the traffic in certain neighbourhoods to decrease ambient noise. Drawing on techniques from \emph{capacitated routing games} \cite{larsson1995augmented,hearn1980bounding} and \emph{constrained MDPs} \cite{altman1999constrained,el2018controlled}, we derive reward modification methods that shifts the game equilibrium mass distribution. Alternatively, constraints may arise in the following scenario: central agent, which we denote by a \emph{social planner}, may enforce constraints to improve user performance as measured by an alternative objective. Equilibria of MDPCGs have been shown to exhibit similar inefficiencies to classic routing games \cite{roughgarden2005selfish,calderonephdthesis}.
As in routing games, we show how reward adjustments can minimize the gap between the equilibrium distribution and the socially optimal distribution \cite{pigou2017economics,cole2006much}.

Since MDPCG models selfish population behaviour under stochastic dynamics, our constraint enforcing methods can be considered as an incentive design framework. One practical application in particular is modifying the equilibrium behaviour of ride-sharing drivers competing in an urban setting. Ride-share has become a significant component of urban mobility in the past decade \cite{furuhata2013ridesharing}. As data becomes more readily available and computation more automated, drivers will have the option of employing sophisticated strategies to optimize their profits---e.g. as indicated in popular media, there are a number of mechanisms available to support strategic decision-making by ride-sharing drivers~\cite{rideshareguy, wired, pbs}. This provides the need for game theoretic models of ride-sharing competition \cite{ahmed2012uncertain}: while rational drivers only seek to optimize their individual profits, ride-sharing companies may choose to incentivize driver behaviours that are motivated by other objectives, such as maintaining driver coverage over large urban areas with varied rider demand as well as increasing overall profits.  

The rest of the paper is organized as follows. Section \ref{sec:related work} provides a discussion of related work. In Section \ref{sec:notation}, we introduce the optimization model of MDPCG's and highlight the relationship between the classical congestion game equilibrium---i.e.~\emph{Wardrop equilibrium}---and Q-value functions from MDP literature. Section \ref{sec:socialPlanner} shows how a social planner can shift the game equilibrium through reward adjustments. Section \ref{sec:dualFormulation} adopts the Frank-Wolfe numerical method~\cite{freund2016new} to solve the game equilibrium and provides an online interpretation of Frank-Wolfe in the context of MDPCG. Section \ref{sec:numerical examples} provides an illustrative application of MDPCG, in which agents repeatedly play a ride-share model game in the presence of population constraints as well as improving the social welfare. Section \ref{conclusion} concludes and comments on future work.

 
\section{Related work}
\label{sec:related work}

Stochastic population games were first studied in the literature as \emph{anonymous sequential games} \cite{jovanovic1988anonymous,bergin1992anonymous,bergin1995anonymous,wikecek2015stationary}.
Recent developments in stochastic population games has been in the \emph{mean field game} \cite{lasry2007mean,gueant2011mean} community. Our work is related to potential mean field games \cite{lasry2007mean,gueant2011infinity} in discrete time, discrete state space \cite{gomes2010discrete} and mean field games on graphs \cite{gueant2014mean,gueant2015existence}. 

Our work can also be thought of as a \emph{continuous population potential game} \cite{Sandholm2001} where the strategies are policies of an MDP or as a modification of classic nonatomic routing games \cite{patriksson2015traffic} where routes have been replaced by policies.   

Techniques for cost modification to satisfy capacity constraints in nonatomic routing games were developed in 
\cite{larsson1995augmented}. See also \cite[Sec. 2.8.2]{patriksson2015traffic} for a discussion of tolling to enforce side-constraints  and \cite[Sec. 2.4]{patriksson2015traffic} for a discussion of tolling to improve social welfare in routing games.

\section{MDP Congestion Games}\label{sec:notation}
We consider a continuous population of selfish agents each solving a finite-horizon MDP with horizon of length $T$, finite state space $\mc{S}$, and finite action space $\mc{A}$. We use the notation $[T]=\{1, \ldots, T\}$ to denote the integer set of length $T$. 

The population mass distribution, $y \in \reals^{T \times |\mc{S}| \times |\mc{A}|}$, is defined for each time step $t \in [T]$, state $s \in \mc{S}$, and action $a\in \mc{A}$.  $y_{tsa}\in \reals$ is the population mass
in state $s$ taking action $a$ at time $t$, and $\sum_a y_{tsa}$ is the total population mass in state $s$ at time $t$. 

Let $P \in \reals^{(T-1) \times |\mc{S}| \times |\mc{S}| \times |\mc{A}|}$ be a stochastic transition tensor. $P_{ts'sa} \in \reals$ defines the probability of $y_{tsa}$ transitioning to state $s'$ in stage $t+1$ when action $a$ is chosen. The transition tensor $P$ is defined such that
\[P_{ts'sa}\geq 0 \quad \forall\ s', s \in \mc{S}, \ a \in \A ,\  t \in [T]\]
and
\[ \sum_{s' \in \mc{S}}\sum_{a \in \A}P_{ts'sa} = 1 \quad \forall \ s \in\mc{S}, \ t \in [T]\]

The population mass distribution obeys the stochastic mass propogation equation
\begin{subequations}
\begin{align}
\sum\limits_{a\in\mathcal{A}}y_{0sa}& =p_s,\quad \forall \ s \in \mc{S} \notag \\
\sum\limits_{a\in\mathcal{A}} y_{t+1, sa} & = \sum\limits_{s'\in\mathcal{S}}\sum\limits_{a\in\mathcal{A}}P_{tss'a}y_{ts'a}, \,\, \forall \ t \in [T-1] \notag 
\end{align}
\end{subequations}
where $p_s$ is the initial population mass in state $s$.

The reward of each time-state-action triplet is given by a function $r_{tsa}: \reals_+^{T \times |\mc{S}| \times |\mc{A}|} \rightarrow \reals$. $r_{tsa}(y)$ is the reward for taking action $a$ in state $s$ at time $t$ for given population distribution $y$. One important case is where $r_{tsa}(y)$ simply depends on $y_{tsa}$, i.e. there exists functions $\ell_{tsa}:\reals_+ \rightarrow \reals$ such that 
\begin{align}
r_{tsa}(y) = \ell_{tsa}(\delta_{tsa}^T y)
\label{rform}
\end{align}
where $\delta_{tsa}$ is an indicator vector for $(t,s,a)$ such that $\delta_{tsa}^T y = y_{tsa}$.  
We say the game is a \emph{congestion game} if the rewards have the form of \eqref{rform} and the functions $\ell_{tsa}(y_{tsa})$ satisfy the following assumption.

\begin{assumption}
\label{assump:decreasing}
$\ell_{tsa}(y_{tsa})$ is a strictly decreasing continuous function of $y_{tsa}$ for each $t,s,a$.
\end{assumption}
Intuitively, the reward of each time-state-action triplet decreases as more members of the population choose that state-action pair at that time.  
We will  use $r(y)$ or $\ell(y)$ to refer to the tensor of all reward functions in each case.

Each member of the population solves an MDP with population dependent rewards  $r_{tsa}(y)$.  As in the MDP literature, we define Q-value functions for each $(t,s,a)$ pair as 
\begin{equation}\label{eq:qvalue}
\begin{split}
Q_{tsa} & = \begin{cases}
r_{tsa}(y) + \underset{s'}{\sum}P_{ts'sa}\Big(\underset{a}{\max}\,\, Q_{t+1,s'a}\Big) & t \in [T-1]\\
r_{tsa}(y)  & t = T
\end{cases}
\end{split}
\end{equation}
In the game context, Q-value function $Q_{tsa}(y)$ represents the distribution dependent payoff that the population receives when choosing action $a$ at $(t,s)$. The Q-value functions can be used to define an equilibrium akin to the Wardrop equilibrium of routing games~\cite{calderone2017markov}. 
\begin{definition}[Wardrop Equilibrium \cite{calderone2017markov}] \label{def:wardrop} A population distribution over time-state-action triplets, 
$\{y^{\star}_{tsa}\}_{t\in [T], s\in \mc{S}, a\in \mc{A}}$
is an \emph{MDP Wardrop equilbrium} for the corresponding MDPCG, if for any $(s,t)$,  $y^{\star}_{ts{a}} > 0 $ implies
\begin{equation}\label{eq:wardrop}
   Q_{ts{a}} \geq Q_{tsa'} \quad \forall a' \neq a, \ a' \in \mc{A} 
\end{equation}
\end{definition}
Intuitively, definition \ref{def:wardrop} amounts to the fact that at every state and time, population members only choose actions that are optimal.

When game rewards satisfy assumption~\ref{assump:decreasing}, MDPCG can be characterized as a \emph{potential game}. 
\begin{definition}[Potential Game \cite{Sandholm2001,calderone2017markov}]
We say that the MDPCG associated with rewards $r(y)$ is a \emph{potential game} if there exists a continuously differentiable function $F$ such that
\begin{align}
\frac{\partial F}{\partial y_{tsa}} = r_{tsa}(y) \notag
\end{align}
\end{definition}
In the specific case when the rewards have form~\eqref{rform}, we can use the potential function
\begin{align}\label{canon}
F(y) =  \sum_{t \in \mc{T}} \sum_{s \in \mc{S}} \sum_{a \in \mc{A}} \int_0^{y_{tsa}} \ell_{tsa}(x) \ dx 
\end{align}
As shown in \cite[Theorem 1.3]{calderone2017markov} given a potential function $F(y)$, the equilibrium to the finite horizon MDPCG can be found by solving the following optimization problem for an initial population distribution $p$. 
\begin{subequations}\label{eq:mdpgame}
\begin{align}
\underset{y}{\max} & \,\,
F(y) \notag\\
\mbox{s.t.} & \,\,\sum\limits_{a\in\mathcal{A}} y_{t+1, sa} = \sum\limits_{s'\in\mathcal{S}}\sum\limits_{a\in\mathcal{A}}P_{tss'a}y_{ts'a}, \,\, \forall \ t \in [T-1], \notag\\ 
& \,\, \sum\limits_{a\in\mathcal{A}}y_{0sa}=p_s,\quad \forall \ s \in \mc{S},\notag\\
& \,\, y_{tsa}\geq 0,\quad \forall \ s\in\mathcal{S}, a\in\mathcal{A}, t\in[T]\notag
\end{align}
\end{subequations}
The proof that the optimizer of~\eqref{eq:mdpgame} is a Wardrop equilibrium relies on the fact that the Q-value functions~\eqref{eq:qvalue} are encoded in the KKT optimality conditions of the problem. The equilibrium condition \eqref{eq:wardrop} is then specifically derived from the complementary slackness condition~\cite{calderone2017markov}. When $F(y)$ has form~\eqref{canon} and Assumption~\ref{assump:decreasing} is satisfied, $F(y)$ is strictly concave, and MDPCG~\eqref{eq:mdpgame} has a unique Wardrop equilibrium.

\section{Constrained MDPCG}\label{sec:convergence}
In this section, we analyze the problem of shifting the game equilibrium by augmenting players' reward functions. In Section~\ref{sec:socialPlanner}, we show that introducing constraints cause the optimal population distribution to obey Wardrop equilibrium for a new set of Q-value functions. Section~\ref{sec:dualFormulation} outlines the Frank Wolfe numerical method for solving a constrained MDPCG as well as provides a population behavioural interpretation for the numerical method. 

\subsection{Planning Perspective: Model and Constraints}\label{sec:socialPlanner}
The Wardrop equilibrium of an MDP congestion game is given by \eqref{eq:mdpgame}. The planner may use additional constraints to achieve auxiliary global objectives. 
For example, in a city's traffic network, certain roads may pass through residential neighbourhoods. A city planner may wish to artificially limit traffic levels to ensure residents' wellbeing. 

We consider the case where the social planner wants the equilibrium population distribution to satisfy constraints of the form
\begin{equation}\label{eq:stateActionConstraints}
g^i(y) \geq 0 \quad \forall i \in \mc{I}
\end{equation}
where $g^i$ are continuously differentiable concave functions.

The social planner cannot explicitly constrain players' behaviour, but rather seeks to add incentive functions $\{f_{tsa}^i\}_{i\in\mc{I}}$ to the reward functions $\ell(y)$ in
order to shift the equilibrium to be within the constrained set defined by \eqref{eq:stateActionConstraints}.  The modified rewards have form  
\begin{equation}\label{eq:generalReward}
\bar{r}_{tsa}(y) = r_{tsa}(y) + \sum\limits_{i \in \mc{I}}f_{tsa}^i(y)
\end{equation}
To determine the incentive functions, the social planner first solves the constrained optimization problem
\begin{subequations}
\begin{align}
\underset{y}{\max} & \,\, F(y) \notag \\
\mbox{s.t.} & \,\, \sum\limits_{a\in\mathcal{A}} y_{t+1, sa} = \sum\limits_{s'\in\mathcal{S}}\sum\limits_{a\in\mathcal{A}}P_{tss'a}y_{ts'a}, \,\, \forall \ t\in[T-1],\notag\\
& \,\,  \sum\limits_{a\in\mathcal{A}}y_{0sa}=p_s, \,\, \forall \ s \in \mc{S}\notag\\
& \,\,  y_{tsa} \geq 0, \,\, \forall \ s\in\mathcal{S},\ a\in\mathcal{A},\ t\in[T]\notag \\
& \,\,  g_i(y) \geq 0, \,\, \forall \ i \in \mc{I} \label{eq:newconstraints}
\end{align}
\label{eq:CMDPCG}
\end{subequations}
and then computes the incentive functions as
\begin{align}\label{eq:tollform}
f_{tsa}^i(y) = (\tau^i)^\star\frac{\partial g^i}{\partial y_{tsa}}(y)
\end{align}
where $\{(\tau^i)^\star \in \reals_+\}_{i \in \mc{I}}$ are the optimal Lagrange multipliers associated with the additional constraints~\eqref{eq:stateActionConstraints}.

The following theorem shows that the Wardrop equilibrium of the MDPCG with modified rewards in~\eqref{eq:generalReward} satisfies the new constraints in~\eqref{eq:newconstraints}.  
\begin{theorem}\label{thm:cWardrop}
Let the MDPCG~\eqref{eq:mdpgame} with rewards $r(y)$ be a potential game with a strictly concave potential function $F(y)$.  If $y^\star$ is a Wardrop equilibrium for a modified MDPCG with reward functions
\begin{align}\label{eq:modcosts}
\bar{r}_{tsa}(y) = r_{tsa}(y) + \sum_{i\in\mc{I}} (\tau^i)^{\star} \frac{\partial g^i}{\partial y_{tsa}}(y)
\end{align}
then $y^\star$ also solves~\eqref{eq:CMDPCG} and thus satisfies the additional constraints~\eqref{eq:stateActionConstraints}.
\end{theorem}
\begin{proof}
The Lagrangian of \eqref{eq:CMDPCG}  is given by
\begin{equation}
\begin{split}
& L(y, \mu, V,\tau) = F(y)-\sum_{tsa}\mu_{tsa}y_{tsa} + \sum\limits_{i} \tau^{i} g^i(y)\\
& + \sum_{t = 1}^{T-1}\sum_s \Big(\sum_{as'} P_{t,ss'a}y_{t,s'a} - \sum_ay_{t+1,sa}\Big)V_{t+1,s} \\
& + \sum_s \Big( p_{s} - \sum_ay_{1sa} \Big)V_{1s}\\
\end{split}
\label{eq:constrained lagrangian}
\end{equation}
and note that by strict concavity
\begin{align}
\sup_{y\geq 0} \inf_{\mu \geq 0,V,\tau \geq 0} \ L(y,\mu,V,\tau) \notag
\end{align}
has unique solution, which we denote by $(y^\star,\mu^\star,V^\star,\tau^\star)$.  
We then note that 
\begin{equation}\label{eq:constrainedObj}
\bar{F}(y) = F(y) + \sum\limits_{i}(\tau^i)^\star g^i(y)
\end{equation}
is a potential function for the MDPCG with modified rewards \eqref{eq:modcosts}. Since $F(y)$ is strictly concave, $g^i(y)$ is concave, and $(\tau^i)^\star$ is positive, $\bar{F}(y)$ is strictly concave. 
The equilibrium for the MDPCG with modified rewards can be computed by solving \eqref{eq:CMDPCG} with $\bar{F}(y)$ as the objective. 

The Lagrangian for \eqref{eq:CMDPCG} with $\bar{F}(y)$ is given by
\begin{align}
\bar{L}(y,\mu,V) = L(y,\mu,V,\tau^\star)
\end{align}
Again by strict concavity
\begin{align}
\sup_{y \geq 0} \inf_{\mu \geq 0,V} \bar{L}(y,\mu,V) =
\sup_{y \geq 0} \inf_{\mu \geq 0,V} L(y,\mu,V,\tau^\star)
\end{align}
has a unique solution which we denote as $(\bar{y}^\star,\bar{\mu}^\star,\bar{V}^\star)$. It follows that $\bar{y}^\star = y^\star$.  Thus the equilibrium of the game with modified rewards, $\bar{y}^\star$ satisfies $g^i(\bar{y}^\star) = g^i(y^\star) \geq 0$ as desired. 
\end{proof}
For the social planner, Theorem~\ref{thm:cWardrop} has the following interpretation: in order to impose constraints of form~\eqref{eq:stateActionConstraints} on a MDPCG, the planner could solve the constrained game~\eqref{eq:CMDPCG} for optimal dual variables $\tau^{\star}$ and offer incentives of form~\eqref{eq:tollform}.  

\subsection{Population Perspective: Numerical Method}\label{sec:dualFormulation}

After the social planner has offered incentives, the population plays the Wardrop equilibrium defined by modified rewards~\eqref{eq:generalReward}; this equilibrium can be computed using the Frank Wolfe (FW) method~\cite{freund2016new}, given in Algorithm~\ref{alg:MSA}, with known optimal variables $\{\tau_i^{\star}\}$. 

FW is a numerical method for convex optimization problems with continuously differentiable objectives and compact feasible sets \cite{powell1982convergence}, including routing games.   One advantage of this learning paradigm is that the population does not need to know the function $r(\cdot)$. Instead, they simply react to the realized rewards of previous game at each iteration. It also provides an interpretation for how a Wardrop equilibrium might be asymptotically reached by agents in MDPCG in an online fashion.

Assume that we have a repeated game play, where players execute a fixed strategy determined at the start of each game. At the end of each game $k$, rewards of game $k$ based on $y^k_{tsa}$ are revealed to all players. FW models the population as having two sub-types: \emph{adventurous} and \emph{conservative}. Upon receiving reward information $\ell_{tsa}(y^{k}_{tsa})$, the adventurous population decides to change its strategy while the conservative population does not. To determine its new strategy, the adventurous population uses value iteration on the latest reward information---i.e.~Algorithm \ref{alg:valueIteration}---to compute a new optimal policy. Their resultant density trajectory is then computed using Algorithm \ref{alg:policy iteration}. The step size at each iteration is equivalent to the fraction of total population who switches strategy.  
The stopping criteria for the FW algorithm is determined by the Wardrop equilibrium notion---that is, as the population iteratively gets closer to an optimal strategy, the marginal increase in potential decreases to zero.
\begin{algorithm}
\caption{Value Iteration Method}
\begin{algorithmic}[h]
\Require \(r\), \(P\).
\Ensure \( \{\pi^{\star}_{ts}\}_{ t \in [T], \, s\in \mc{S}} \)
\For{\(t=T, \ldots, 1\)}
	\ForEach{\(s\in\mathcal{S}\)}
		\State{\(\displaystyle V_{ts}= \underset{a\in\mathcal{A}}{\mbox{max}}\,\, Q_{tsa}\)} \Comment{Eqn.~\eqref{eq:qvalue}}
		\State{\(\displaystyle \pi^{\star}_{ts} = \underset{a\in\mathcal{A}}{\mbox{argmax}}\,\, Q_{tsa}\)}  \Comment{Eqn.~\eqref{eq:qvalue}}
	\EndFor 
\EndFor 
\end{algorithmic}
\label{alg:valueIteration}
\end{algorithm}
\begin{algorithm}
\caption{Retrieving density trajectory from a policy }
\begin{algorithmic}[h]
\Require \(P\), \(p\), \(\pi\).
\Ensure \(\{d_{tsa}\}_{t \in[T], s\in\mc{S}, a\in\mc{A}}\)
\State{\(d_{tsa} = 0, \ \forall \ t \in [T], s \in \mc{S}, a \in \mc{A}\)}
\For{\(t=1, \ldots, T\)}
	\ForEach{\(s\in\mathcal{S}\)}
	    \If{\(t = 1\)}
	        \State{\(d_{1s\pi_{1s}} = p_s\)}
	    \Else
		    \State{\(d_{ts(\pi_{ts})} = \sum\limits_{a \in \mc{A}}\sum\limits_{s' \in \mc{S}} P_{t-1,ss'a}d_{t-1,s'a} \)}
		\EndIf
	\EndFor 
\EndFor 
\end{algorithmic}
\label{alg:policy iteration}
\end{algorithm}
\begin{algorithm}
\caption{Frank Wolfe Method with Value Iteration}
\begin{algorithmic}[h]
\Require \(\bar{\ell}\), \(P\), \(p_s\), \(N\), \(\epsilon\).
\Ensure \(\{y^{\star}_{tsa}\}_{t \in [T], s\in\mc{S}, a \in \mc{A}}\).
\State{\( y^0 = 0 \in \reals^{T \times|\mc{S}| \times |\mc{A}|} \)}
\For{\(k = 1, 2, \ldots, N\)}
    \State{\(c^k_{tsa} = \bar{\ell}_{tsa}(y^k)\), \(\quad \forall \,\, t \in [T], s \in \mc{S}, a \in \mc{A}\)}
	\State{$\pi_{ts} = $ \Call{ValueIteration}{$c^k$, $ P$}}
	\Comment{Alg.~\ref{alg:valueIteration}}
	\State{\(d^k\) = \Call{RetrieveDensity}{$P,\ p_s,\ \pi_{ts}$}}\Comment{Alg.~\ref{alg:policy iteration}}
	\State{\( \alpha^k =  \frac{2}{k+1} \)}
	\State{\( y^k = (1 - \alpha^k)y^{k-1} + \alpha^k d^k\)}
	\State{Stop if }
	\State{\(\quad \sum\limits_{t\in[T]} \sum\limits_{s\in\mc{S}} \sum\limits_{a\in\mc{A}}\Big(c^k_{tsa} - c^{k-1}_{tsa}\Big)^2 \leq \epsilon\)}
\EndFor
\end{algorithmic}
\label{alg:MSA}
\end{algorithm}

In contrast to implementations of FW in routing game literature, Algorithm \ref{alg:MSA}'s descent direction is determined by solving an MDP \cite[Section~4.5]{puterman2014markov} as opposed to a shortest path problem from origin to destination \cite[Sec.4.1.3]{patriksson2015traffic}. 
Algorithm~\ref{alg:MSA} is guaranteed to converge to a Wardrop equilibirum if the predetermined step sizes decrease to zero as a harmonic series \cite{freund2016new}---e.g., $\frac{2}{k+1}$.
FW with predetermined step sizes has been shown to have sub-linear worst case convergence in routing games \cite{powell1982convergence}. 
On the other hand, replacing fixed step sizes with optimal step sizes found by a line search method leads to a much better convergence rate.

\section{Numerical example}\label{sec:numerical examples}




In this section, we apply the techniques developed in Section \ref{sec:convergence} to model competition among ride-sharing drivers in metro Seattle.  
Using the set up described in Section \ref{sec:rideshare model}, we demonstrate how a ride-share company takes on the role of social planner and shifts the equilibrium of the driver game in the following two scenarios: 
\begin{list}{--}{}
    \item Ensuring minimum driver density in various neighborhoods (Section \ref{sec:lowerbound downtown}). 
    \item Improving the social welfare (Section \ref{sec:braess paradox}). 
\end{list}

\subsection{Ride-sharing Model}
\label{sec:rideshare model}
\begin{figure}
\centering
\includegraphics[width=1.0\columnwidth]{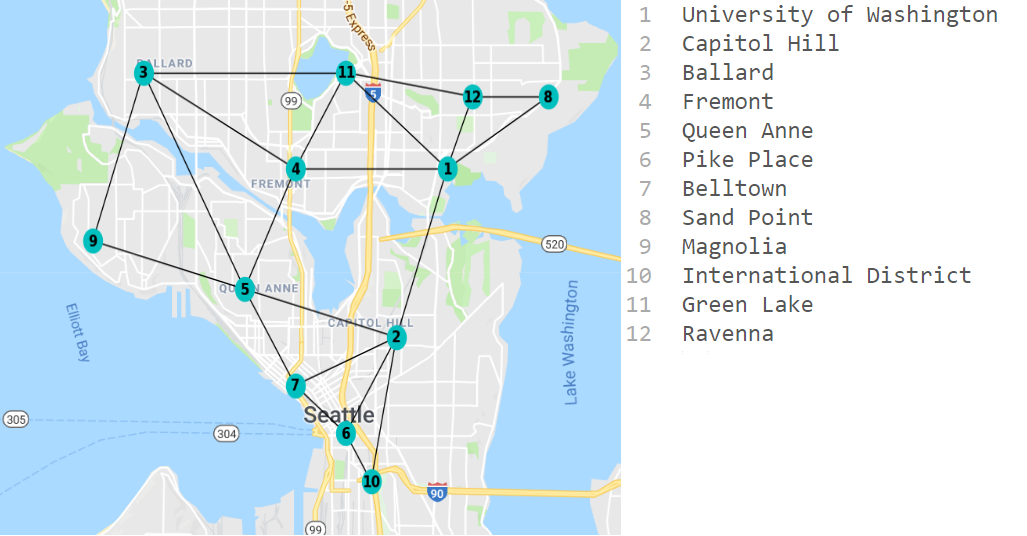}
\caption{State representation of metro Seattle.}
\label{fig:nbhds}
\end{figure}
Consider a ride-share scenario in metro Seattle, where rational ride-sharing drivers seek to optimize their profits while repeatedly working Friday nights. Assume that the demand for riders is constant during the game and for each game play. The \emph{time step} is taken to be 15 minutes, i.e. the average time for a ride, after which the driver needs to take a new action.  

We model Seattle's individual neighbourhoods as an abstract set of \emph{states}, $s \in \mc{S}$, as shown in Fig \ref{fig:nbhds}. Adjacent neighbourhoods are connected by edges. The following states are characterized as residential: `Ballard' (3), `Fremont' (4), `Sand Point' (8), `Magnolia' (9), `Green Lake' (11), `Ravenna' (12). Assume drivers have equal probabilities of starting from any of the residential neighbourhoods.

Because drivers cannot see a rider's destination until after accepting a ride, the game has MDP dynamics. At each state $s$, drivers can choose from two actions.  $a_r$, wait for a rider in $s$, or $a_{s_j}$, transition to an adjacent state $s_j$. When choosing $a_r$, we assume the driver will eventually pick up a rider, though it may take more time if there are many drivers waiting for riders in that neighborhood. Longer wait times decrease the driver's reward for choosing $a_r$. 

On the other hand, there are two possible scenarios when drivers choose $a_{s_j}$. The driver either drives to $s_j$ and pays the travel costs without receiving a fare, or picks up a rider in $s_i$. We allow the second scenario with a small probability to model the possibility of drivers deviating from their predetermined strategy during game play.

The \emph{probability of transition} for each action at state $s_i$ are given in \eqref{eq:probs}. $N_i$ denotes the set of neighbouring states, and $|N_i|$ the number of neighbouring states for state $s_i$. 
\begin{equation}
    P(s,a,s_i)=
    \left\{\begin{array}{ll}
  \frac{1}{|N_i|+1}, & \text{if}\ s\in N_i, a=a_r\\
  \frac{1}{|N_i|+1}, & \text{if}\ s = s_i, a=a_r\\
  \frac{0.1}{|N_i|}, & \text{if}\ s\in N_i, s \neq s_j, a=a_{s_j}\\
  0.9, & \text{if}\ s \in N_i, s=s_j, a=a_{s_j}\\ 
  0,& \text{otherwise}
  \end{array}\right.
  \label{eq:probs}
\end{equation}
The \emph{reward function} for taking each action is given by 
\begin{align}
\ell_{tsa}(y_{tsa}) & =\mathds{E}_{s'} \left[m_{ts's} - c_{ts's}^\text{trav}\right] - c_{t}^\text{wait}\cdot y_{tsa} \notag \\
& = \sum_{s'}P_{ts'sa}\left[m_{ts's} - c_{ts's}^\text{trav}\right] - c_{t}^\text{wait}\cdot y_{tsa} \notag
\end{align}
where $m_{ts's}$ is the monetary cost for transitioning from state $s$ to $s'$, $c_{ts's}^\text{trav}$ is the travel cost from state $s$ to $s'$, $c_{t}^\text{wait}$ is the coefficient of the cost of waiting for a rider.  We compute these various parameters as
\begin{subequations}
\begin{align}
m_{ts's} & = \big(\text{Rate}\big) \cdot \big(\text{Dist}\big) \\
c_{ts's}^\text{trav} & = 
\tau 
\underbrace{\big(\text{Dist}\big)}_{\text{mi}}
\underbrace{\big(\text{Vel}\big)^{-1}}_{\text{hr}/\text{mi}} + 
\underbrace{\big(\substack{\text{Fuel} \\ \text{Price} } \big)}_{\$/\text{gal}}
\underbrace{\big( \substack{\text{Fuel} \\ \text{Eff}} \big)^{-1}}_{\text{gal}/\text{mi}} 
\underbrace{\big(\text{Dist}\big)}_{\text{mi}} \\
c^\text{wait}_{ta} & = 
\begin{cases}
\tau \cdot 
\Big(
\underbrace{
\substack{\text{Customer} \\ \text{Demand Rate} }
}_{\text{rides}/\text{hr}}
\Big)^{-1}, & \text{ if $a=a_r$} \\
\epsilon_{tsa_{s'}}, & \text{ if $a=a_s'$}
\end{cases}
\label{eq:travel}
\end{align}
\end{subequations}
where $\epsilon_{tsa_{s'}}$ is the congestion effect from drivers who all decide to traverse from $s$ to $s'$, and $\tau$ is a time-money tradeoff parameter, computed as 
$$\frac{\text{Rate}\cdot {D_{ave}}}{\text{Time Step}}$$
Where the average trip length, ${D_{ave}}$, is equivalent to the average distance between neighbouring states. The values independent of specific transitions are listed in the Tab.~\ref{tab:params}. 
\begin{table}[h!]
\begin{center}
\begin{tabular}{|cccccc|}
\hline
Rate  & Velocity & Fuel Price & Fuel Eff & $\tau$ & ${D_{ave}}$   \\
\hline
\$6 /mi & 8 mph & \$2.5/gal & 20 mi/gal & \$27 /hr  & 1.25 mi\\
\hline
\end{tabular}
\end{center}
\caption{Parameters for the driver reward function.}\label{tab:params}
\end{table}
\subsection{Ensuring Minimum Driver Density}
\label{sec:lowerbound downtown}
To ensure rider satisfaction, the ride-share company aims for a minimum driver coverage of 10 drivers in `Belltown', $s = 7$, a neighborhood with highly variable rider demand. 
To this end, they solve the optimization problem in \eqref{eq:CMDPCG} where \eqref{eq:newconstraints} for $t\in \{3, \ldots, T\}$, $s = 7$, take on the form
\begin{equation}
\begin{aligned}
g_{i}(y) = \sum_{a} y_{tsa} -  10 
\end{aligned}
\label{eq:downtown constraint}
\end{equation}
The modified rewards~\eqref{eq:modcosts} are given by
\begin{align}
\bar{r}_{tsa}(y) & = \ell_{tsa}(y_{tsa}) + \tau^\star_{ts}
\label{eq:generalized cost}
\end{align}
where each $\tau^\star_{ts}$ is the optimal dual variable corresponding to each new constraint.  

The optimal population distribution in `Belltown' (state 7) and  an adjacent neighbourhood, `Capitol Hill' (state 2), are shown in Fig.~\ref{fig:unconstrained solution}.
\begin{figure}
\center
\includegraphics[trim={0 4cm 0 0}, width=1.0\columnwidth] {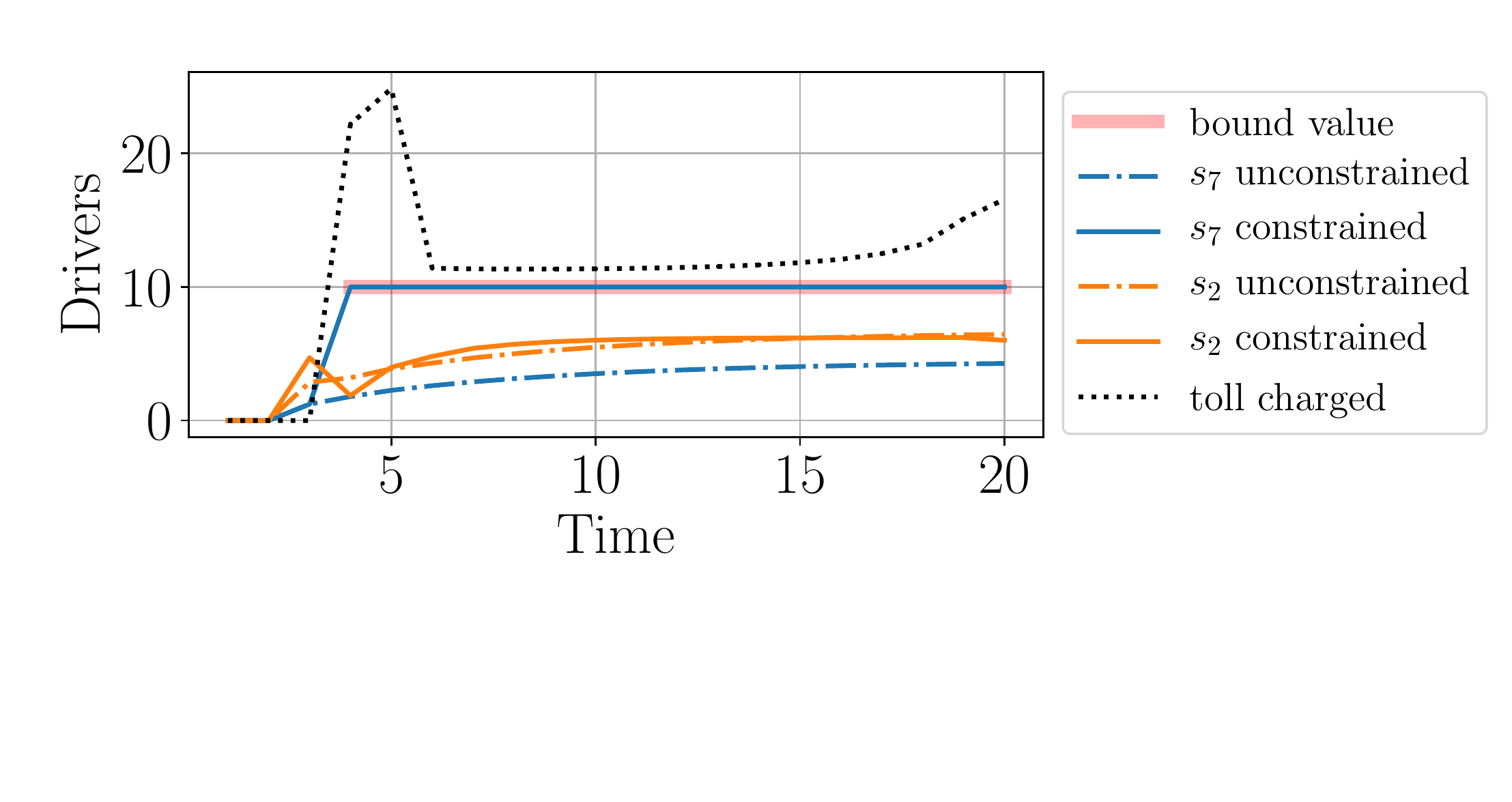}
\caption{State density of the optimal trajectory solution to \eqref{eq:CMDPCG}. A constraint of form \eqref{eq:downtown constraint} is placed on `Belltown', for $t \in [3,20]$. The imposed constraints also affect optimal population distribution of other states, as shown by changes in the population distribution of neighboring state $s = 2$. }
\label{fig:unconstrained solution}
\end{figure}
Note that the incentive $\tau^\star_{ts}$ is applied to all actions of state $s$. Furthermore, if the solution to the unconstrained problem is feasible for the constrained problem, then $\tau^{\star}_{ts} = 0$---i.e.~no incentive is offered. 
We simulate drivers' behaviour with Algorithm \ref{alg:MSA}, as a function of decreasing termination tolerance $\epsilon $. In Fig.~\eqref{fig:MSA simulation}, the result shows that the optimal population distribution from the FW algorithm converge to Wardrop equilibrium as the approximation tolerance $\epsilon$ decreases.

\begin{figure}
\center
\subfloat{\includegraphics[trim={3cm 1cm 1cm 1cm},width=0.7\columnwidth]{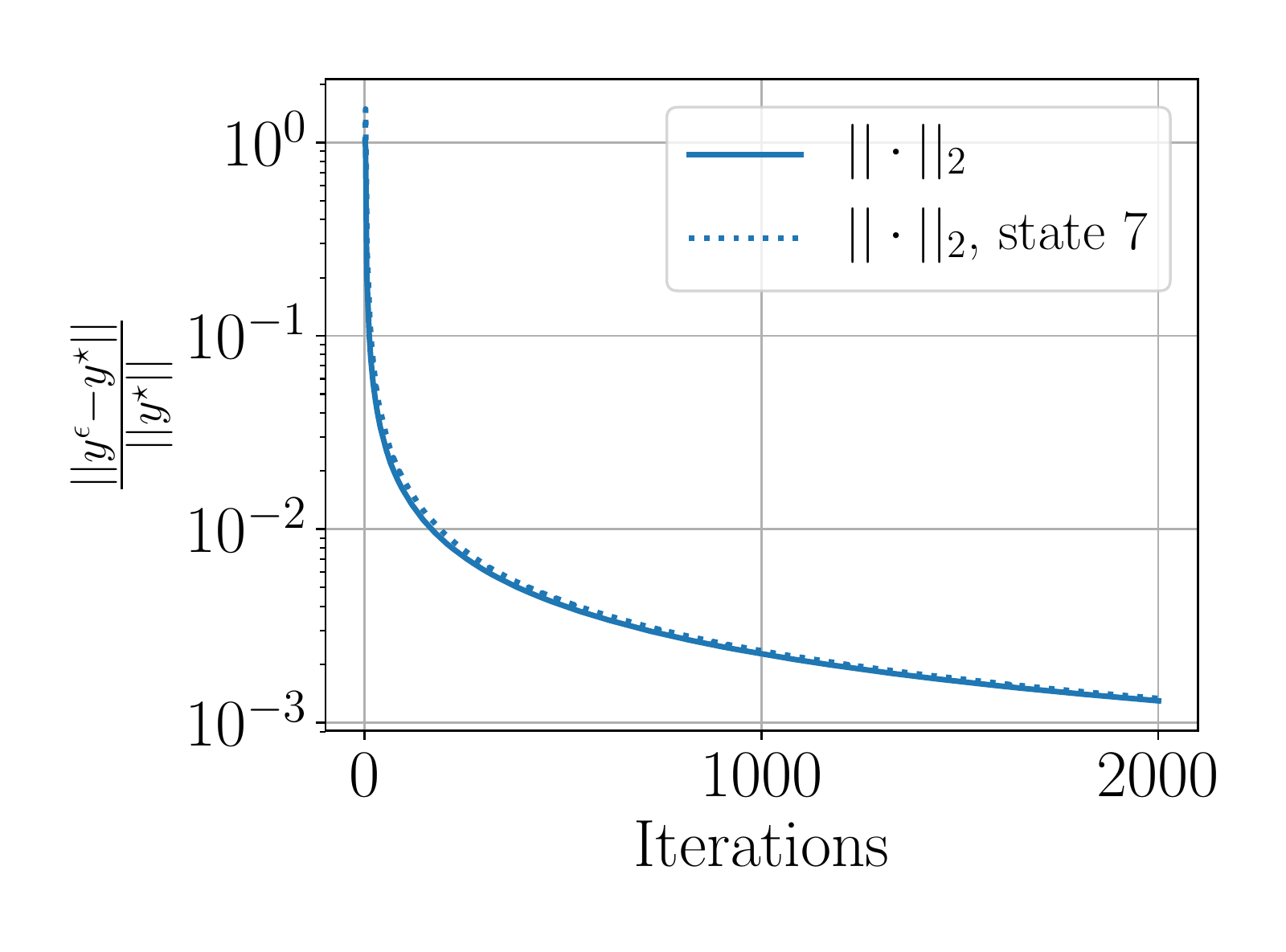}
} 
\caption{Convergence of $y^\epsilon$ to $y^\star$. Plot shows difference between constrained optimal solution and FW approximation, $\norm{y^{\epsilon} - y^{\star}}_2 $, normalized by $\norm{y^{\star}}_2$. }
\label{fig:MSA simulation}
\end{figure}

\subsection{Increasing Social Welfare}
\label{sec:braess paradox}
In most networks with congestion effects, the population does not achieve the maximum \emph{social welfare}, which can be achieved by optimizing~\eqref{eq:mdpgame} with objective
\begin{equation}
J(y) = \sum\limits_{t\in[T]}\sum\limits_{s\in\mathcal{S}} \sum\limits_{a\in\mathcal{A}}y_{tsa}r_{tsa}(y)
\label{eq:social objective}
\end{equation} 
In general, a gap exists between $J(x^{\star})$ and $J(y^{\star})$, where $y^\star=\{y^\star_{tsa}\}$ is the optimal solution to \eqref{eq:mdpgame}, and $x^\star=\{x^\star_{tsa}\}$ is the optimal solution to \eqref{eq:mdpgame} with objective \eqref{eq:social objective}. 

The typical approach to closing the social welfare gap is to impose mass dependent incentives. An alternative method, perhaps under-explored, is to impose constraints. As opposed to congestion dependent taxation methods for improving social welfare~\cite{pigou2017economics,beckmann1952continuous}, constraint generated tolls are congestion independent. 

We can compare the two distributions and generate upper/lower bound constraints with an $\epsilon$ threshold---see Algorithm~\ref{alg:constraint gen} for constraint selection method. The number of constraints increases with decreasing $\epsilon$. Since the objective function in \eqref{eq:mdpgame} is continuous in $y_{tsa}$, as $\epsilon$ approaches zero, the objective will also approach the social optimal. 
\begin{algorithm}
\caption{Constraint Generation}
\begin{algorithmic}[h]
\Require \(x^\star\), \(y^\star\).
\Ensure \(\mc{U} = \{ ( u_i, t,s,a) \in \reals \times [T] \times \mc{S}\times\mc{A}\}\)
\State{    \(\quad\quad\mc{L} = \{ (l_i, t,s,a) \in \reals \times [T] \times \mc{S}\times\mc{A}\}\)}
\ForEach{\(s\in\mathcal{S}\), \(a\in\mathcal{A}\), \(t\in[T]\)}
	\If{\(y^\star_{tsa} - x^\star_{tsa} > \epsilon \)}
		\State{\((x^\star_{tsa}, t,s,a) \rightarrow \mc{U} \)}
	\ElsIf{\(y^\star_{tsa} - x^\star_{tsa} < \epsilon\)}
		\State{\((x^\star_{tsa}, t,s,a) \rightarrow \mc{L} \)}
	\EndIf
\EndFor
\end{algorithmic}
\label{alg:constraint gen}
\end{algorithm}
\begin{figure}
\center
\subfloat{\includegraphics[trim={1.4cm 1cm 0 0.5cm},width=0.37\textwidth] {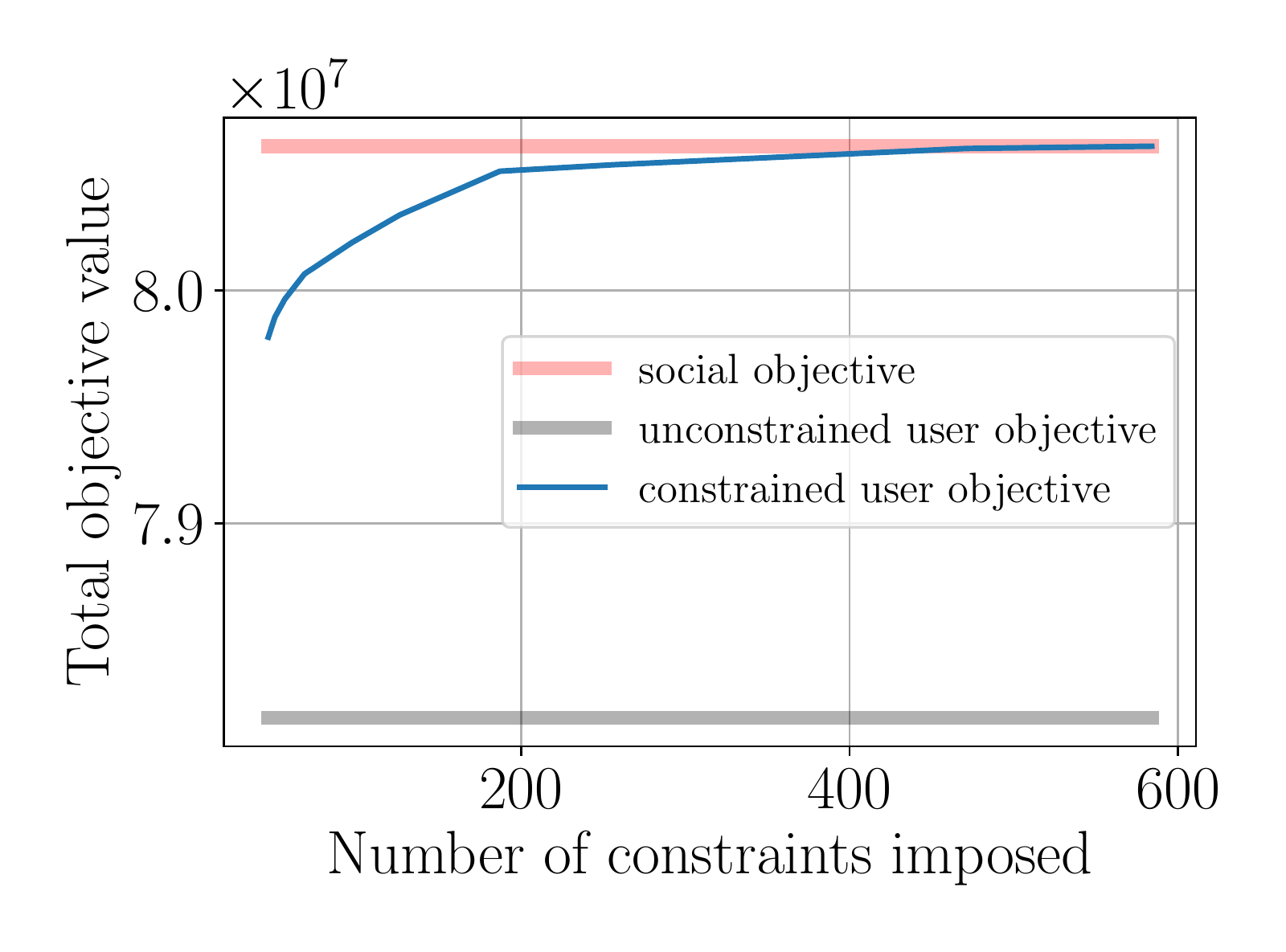}} 
\caption{With population of $3500$, the social welfare of the user selected equilibrium is shown as a function of the number of imposed constraints; increasing the number of constraints is equivalent to  decreasing $\epsilon$ tolerance.}
\label{fig:braessParadox}
\end{figure}
\begin{figure}
\center
\subfloat{\includegraphics[trim={1.4cm 1.5cm 0 0}, width=0.37\textwidth]{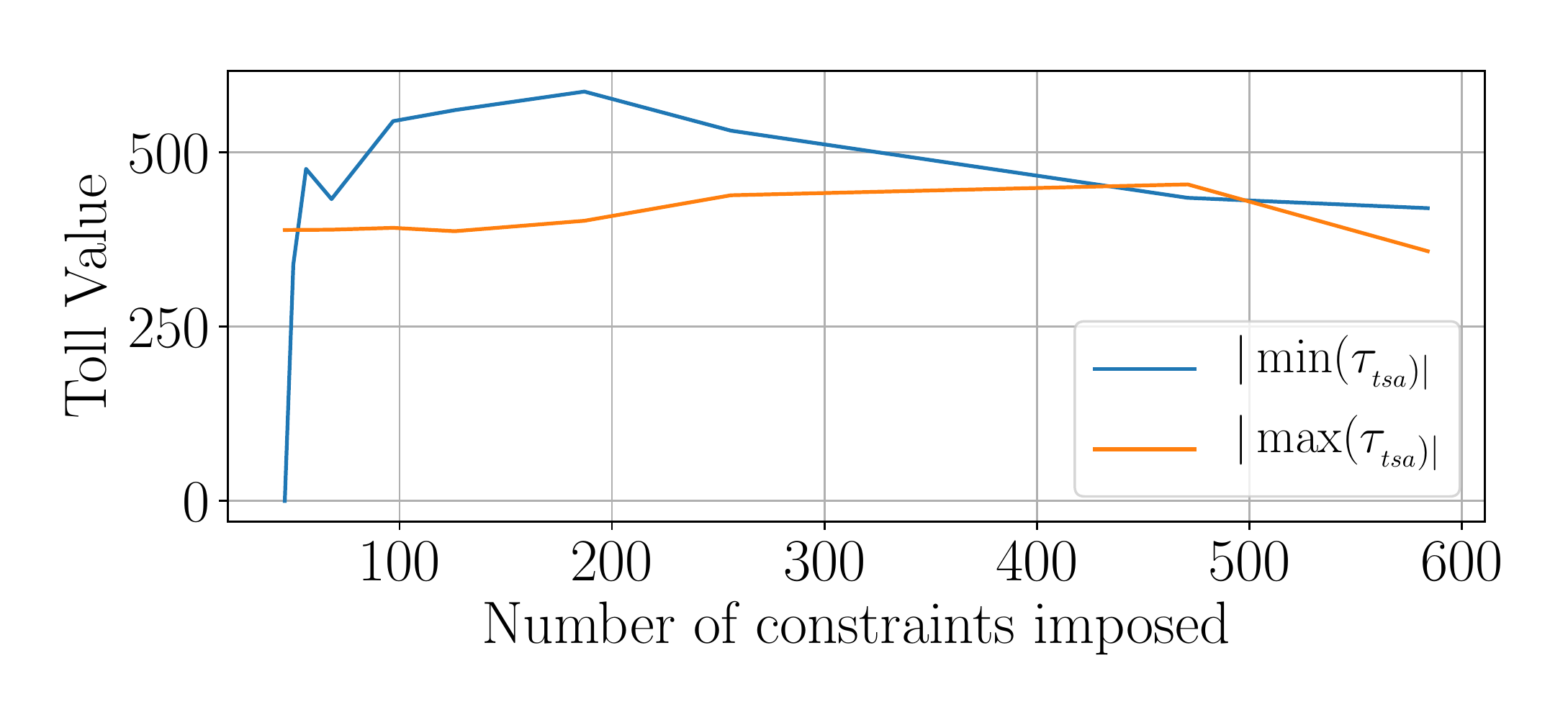}} \newline
\subfloat{\includegraphics[trim={1cm 4cm 0 0}, width=0.45\textwidth]{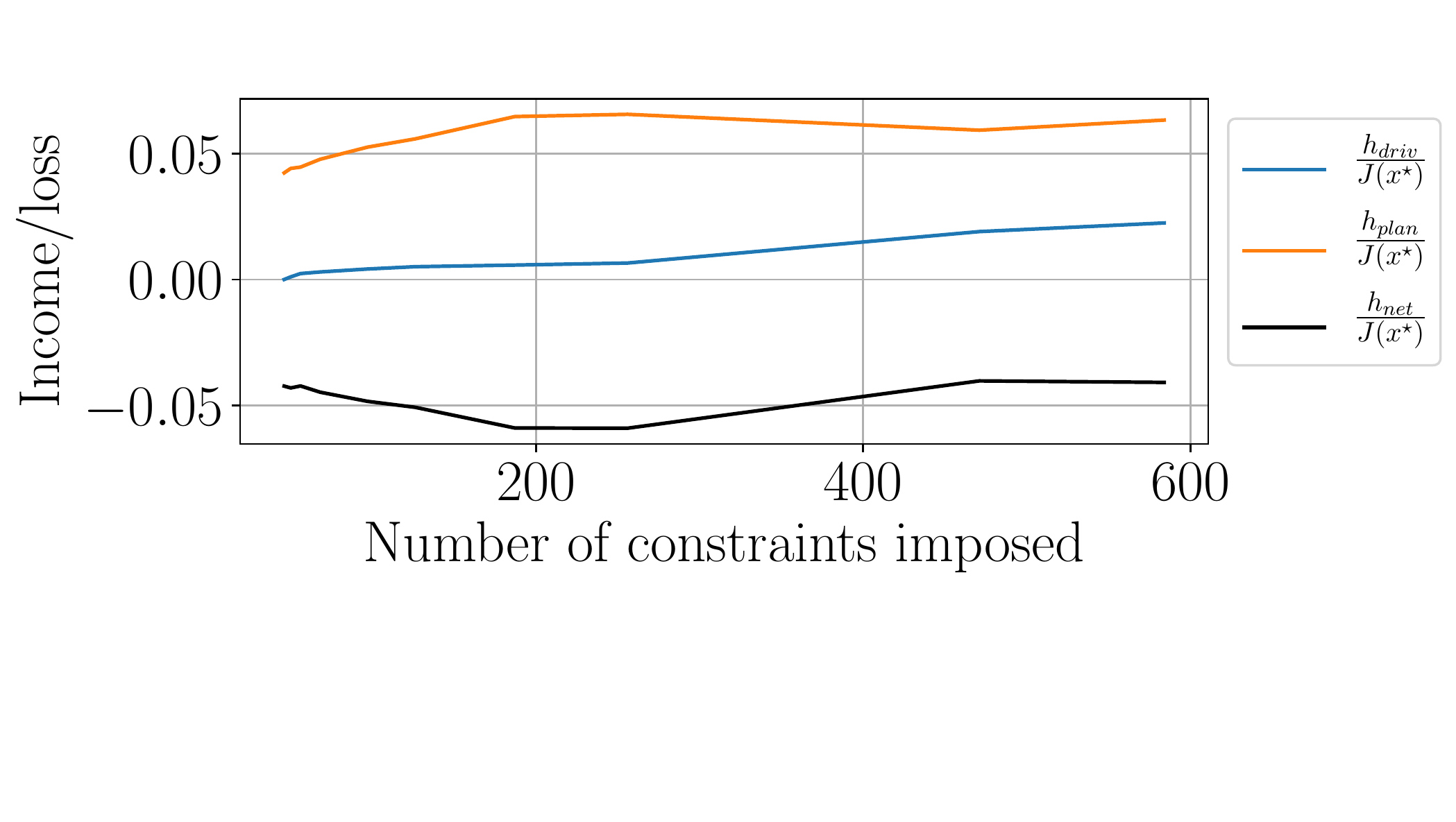}} 
\caption{Maximum and minimum toll values are shown in (a) as a function of number of constraints. In (b), the income/loss required to increase social welfare is shown as a function of constraints imposed.}
\label{fig:tollAnalysis}
\end{figure}
In Fig.~\ref{fig:braessParadox}, we compare the optimal social welfare to the social welfare at Wardrop equilibrium of the unconstrained congestion game, modeled in Section \eqref{sec:rideshare model}, as a function of the population size. We use CVXPY \cite{diamond2016cvxpy} to solve the optimization problem. 
 
We utilize Algorithm~\ref{alg:constraint gen} to generate incentives for the congestion game. Then, we simulate~\eqref{eq:CMDPCG} and compare the game output to the social objective in Fig.~\ref{fig:braessParadox}.  For a population size of $3500$, there is a discernible gap between the social and user-selected optimal values.  Note that with only $200$ $(t,s,a)$ constraints, the gap between the social optimal and the user-selected equilibrium is already less than $5\%$.

An interesting question to ask is how much of the total market worth is affected by the incentives. In Fig.~\ref{fig:tollAnalysis}, we demonstrate how payouts vary based on the number of constraints imposed. Let $(\cdot)_- = \min\{0, \cdot\}$ and $(\cdot)_+ = \max\{0, \cdot\}$.  The total payout from the drivers to the social planner and the social planner to the drivers are given by 
\begin{align*}
h_\text{driv} & = \sum_{tsa}y_{tsa}|(\tau_{tsa})_-|, \quad h_\text{plan} & =  \sum_{tsa}y_{tsa}(\tau_{tsa})_+
\end{align*}
The net revenue the social planner receives from tolls is
$$
h_\text{net} = \sum_{tsa} y_{tsa}(\tau_{tsa}) = h_\text{plan} - h_\text{driv}.
$$
Fig.~\ref{fig:tollAnalysis}(b) shows how these quantities change as the total number of constraints is increased.

\section{Conclusions}\label{conclusion}
We presented a method for adjusting the reward functions of a MDPCG in order to shift the Wardrop equilibrium to satisfy population mass constraints. Applications of this constraint framework have been demonstrated in a ride-share example in which a social planner aims to constrain state densities or to maximize overall social gain without explicitly constraining the population. Future work include developing online methods that updates incentives corresponding to constraints while the game population adjusts its strategy.

\bibliographystyle{IEEEtran}
\bibliography{reference}

\end{document}